\documentclass[letterpaper, 10 pt, conference]{ieeeconf}
\IEEEoverridecommandlockouts
\usepackage{times}
\usepackage{amsmath,amssymb,amsfonts}
\usepackage{graphicx}
\usepackage{microtype}
\usepackage{booktabs}
\usepackage{xcolor}
\usepackage{mathtools}
\usepackage{cite}
\renewcommand{\QED}{\QEDopen}
\usepackage{algorithm}
\usepackage{algorithmic}

\newtheorem{assumption}{Assumption}
\newtheorem{lemma}{Lemma}
\newtheorem{theorem}{Theorem}

\title{\LARGE \bf
 An Online Learning Approach
for Two-Player Zero-Sum Linear Quadratic Games
}

\author{%
Shanting Wang$^{1}$, Weihao Sun$^{1}$, and Andreas A. Malikopoulos$^{2}$, \textit{Senior Member, IEEE}
\thanks{This research was supported in part by NSF under Grants CNS-2401007, CMMI-2348381, IIS-2415478, and in part by MathWorks.}%
\thanks{$^{1}$S. Wang and W. Sun are with the Systems Engineering, Cornell University, Ithaca, NY, USA. {\tt\small sw997@cornell.edu, ws493@cornell.edu}}%
\thanks{$^{2}$A. A. Malikopoulos is with the Applied Mathematics, Systems Engineering, Mechanical Engineering, Electrical \& Computer Engineering, and School of Civil \& Environmental Engineering, Cornell University, Ithaca, NY, USA. {\tt\small amaliko@cornell.edu}}%
}

\begin{document}
\maketitle
\thispagestyle{empty}
\pagestyle{empty}

\begin{abstract}
In this paper, we present an online learning approach for two-player zero-sum linear quadratic games with unknown dynamics. We develop a framework combining regularized least squares model estimation, high probability confidence sets, and surrogate model selection to maintain a regular model for policy updates. We apply a shrinkage step at each episode to identify a surrogate model in the region where the generalized algebraic Riccati equation admits a stabilizing saddle point solution. We then establish regret analysis on algorithm convergence, followed by a numerical example to illustrate the convergence performance and verify the regret analysis.
\end{abstract}


\section{Introduction}
Dynamic games provide a natural framework for cyber-physical systems in which multiple decision makers interact through a shared dynamical process \cite{Ye2023gamesurvery,Xu2025when}. 
Among such models, linear quadratic (LQ) games are especially important because they retain a useful analytical structure inherited from the classical linear quadratic regulator (LQR) for a single agent \cite{lewis2012optimal}, where the optimal feedback laws are characterized by Riccati equations. 
These properties make the LQ zero-sum games a useful framework in a variety of application domains, including multi-robot coordination \cite{ren2025chance,le2023multirobot}, autonomous driving \cite{liu2023potential}, and security-critical control applications~\cite{xu2025game}, where agents with competing objectives interact strategically \cite{chremos2022CSMArticle}. Within this class, two-player zero-sum LQ games are of particular interest because the objectives of the two players are exactly opposed \cite{bacsar2008h}. At the same time, their equilibrium strategies remain analytically tractable since they can be obtained through the solution to generalized algebraic Riccati equations (GARE).

When the system dynamics are unknown, online learning must be carried out from observed data and controls \cite{Malikopoulos2022a,Malikopoulos2024,kounatidis2025combined}. In the zero-sum games, in particular, two players seek to improve their feedback strategies directly from interaction data without relying on the exact model \cite{ioannou1996robust}. In this setting, regret is a natural performance metric as it measures the deviations in performance from optimality over time, hence it can capture both transient learning efficiency and long-run control performance \cite{dean2018regret, mania2019certainty,simchowitz2020naive}.  In prior research efforts, regret guarantees for online LQR in single-agent settings are established \cite{dean2018regret, mania2019certainty,simchowitz2020naive}. Other studies have investigated approaches to derive optimal strategies for LQ games, using policy optimization \cite{zhang2019policy}, Q-learning \cite{al2007model}, and other data-driven iterative methods \cite{nortmann2024nash}. However, applying online regret minimization to an unknown zero-sum LQ game still lacks full exploration and understanding.

In this paper, we study an online learning approach for two-player zero-sum LQ games with unknown dynamics. In contrast with the classical single-agent LQR setting, this problem is more challenging because the strategies of the two players are coupled through the equilibrium structure, the interaction is adversarial, and the stability must be maintained throughout the learning process. We first apply a least-squares system identification method combined with a confidence region, and then propose a model-shrinkage step to identify a regular surrogate model that preserves stability. The corresponding GARE is solved on the selected model. Finally, we establish a sublinear regret bound based on the policy gap caused by estimation error, exploration error, and the transient bound. 
 
The main contributions of this paper include the following. We develop a certified online learning algorithm for a two-player zero-sum LQ game with unknown dynamics. We establish a sublinear regret bound of $O\!\left(\sqrt{T}\right)$ together with convergence properties of the proposed algorithm. We then provide the numerical results to illustrate that both parameter estimation and feedback policy updates work stably over time.

The remainder of the paper is organized as follows. In Section~\ref{sec:problem}, we present the problem formulation, including the system model, equilibrium, and regret. In Section~\ref{sec:algorithm}, we develop the certified online learning algorithm. In Section~\ref{sec:regret}, we provide the regret bound and convergence analysis. In Section~\ref{sec:simulation}, we present the numerical simulations to show the performance of the algorithm. Finally, in Section~\ref{sec:conclusion}, we draw concluding remarks and discuss potential future directions.

\section{Problem Formulation}
\label{sec:problem}

\subsection{Notation}
Let $\mathbb{R}$ denote the set of real numbers; $\mathbb{R}^n$ denotes the space of real column vectors of dimension $n$, and $\mathbb{R}^{n \times m}$ denotes the space of real matrices of dimension $n \times m$.
Let $I_n\in\mathbb{R}^{n\times n}$ denote the identity matrix. Denote by $(\cdot)^\top$, $\det(\cdot)$, $\otimes$ the transpose, determinant, and the Kronecker product. 
For a matrix $A\in\mathbb R^{n\times n}$, $A\succ0$ indicates $A$ is positive definite. $A \succeq0$ indicates $A$ is positive semidefinite.
$\|x\|_p$ denote the $p$-norm of a vector $x\in\mathbb R^n$. For a positive definite matrix $A\in\mathbb R^{n\times n}$, the weighted 2-norm of the vector $x\in\mathbb R^n$ is defined by $\|x\|_A=\sqrt{x^\top Ax}$. 
Let \(\operatorname{vec}:\mathbb{R}^{n\times d}\to\mathbb{R}^{nd}\) denote the column-wise vectorization operator. For a smooth map $f$ between normed spaces,
$D_x f[h]$ denotes the Fr\'{e}chet derivative of $f$
with respect to $x$ in direction $h$.

We consider a discrete-time linear system where the dynamics are characterized as
\begin{equation}
  x_{t+1} = A x_t + B_1 u_t + B_2 v_t + w_t,
  \label{eq:dynamics}
\end{equation}
where $x_t \in \mathbb{R}^n$ is the state, $u_t \in \mathbb{R}^{m_1}$ and $v_t \in \mathbb{R}^{m_2}$ are the control inputs of players 1 and 2, respectively, and
$A \in \mathbb{R}^{n \times n}$, $B_1 \in \mathbb{R}^{n \times m_1}$, $B_2 \in \mathbb{R}^{n \times m_2}$ are unknown constant matrices.
$x_0 \in \mathbb{R}^n$ is the initial state from a distribution $\mathcal{D}$.
The disturbance $\{w_t\}_{t \ge 0}$ is an $\{\mathcal{F}_t\}$-adapted
martingale difference sequence that is conditionally
$\sigma_w^2$-sub-Gaussian, with covariance
$\Sigma_w \coloneqq \mathbb{E}[w_t w_t^\top]$. 
The objective of player 1 (player 2) is to minimize (maximize) the infinite-horizon average cost function,
\begin{equation}
J(u, v)
:=
\limsup_{T\to\infty}\frac{1}{T}\,
\mathbb{E}\!\left[\sum_{t=0}^{T-1}c_t(x_t,u_t,v_t)\right],
\label{eq:avg_objective}
\end{equation}
where $c_t(x_t,u_t,v_t)
  = x_t^\top Q x_t + {u_t}^\top R_u u_t - {v_t}^\top R_v v_t$, $Q\in\mathbb{R}^{n\times n} \succeq 0$, $R_u\in\mathbb{R}^{m_1\times m_1} \succ 0$ and
$R_v\in\mathbb{R}^{m_2\times m_2}\succ 0$. 

If the $\inf$ and $\sup$ in \eqref{eq:avg_objective} can be interchanged, $J_\star$ is the value of the game at the Nash equilibrium defined by
\begin{equation}
{
  J_\star := \inf_{u}\ \sup_{v}\ J(u, v)
        = \sup_{v}\ \inf_{u}\ J(u, v).
}
\label{eq:value}
\end{equation}

We denote $P_\star$ as a solution to the GARE \cite{zhang2019policy} as:
\begin{equation}
\begin{aligned}
    P_* &= Q + A^\top P_* A \\
    & - \begin{bmatrix} A^\top P_* B_1 & A^\top P_* B_2 \end{bmatrix}
    H(P_*)^{-1}
    \begin{bmatrix} B_1^\top P_* A \\ B_2^\top P_* A \end{bmatrix},
    \label{eq:GARE}
\end{aligned}
\end{equation}
with
\begin{equation}
    H(P) = \begin{bmatrix} 
        R_u + B_1^\top P B_1 & B_1^\top P B_2 \\ 
        B_2^\top P B_1 & -R_v + B_2^\top P B_2 
    \end{bmatrix}.
    \label{eq:H}
\end{equation}

There exists a pair of linear feedback stabilizing policies

\begin{equation}
    u_t^* = -K_* x_t, \qquad v_t^* = -L_* x_t,
    \label{eq:NE-policy}
\end{equation}
where the NE gains $K_\star \in \mathbb{R}^{m_1 \times n}$ and 
$L_\star \in \mathbb{R}^{m_2 \times n}$ are given by
\begin{equation}
    \begin{bmatrix} K_\star \\ L_\star \end{bmatrix} 
    = H(P_\star)^{-1}
    \begin{bmatrix} B_1^\top P_\star A \\ B_2^\top P_\star A \end{bmatrix}.
    \label{eq:NE-gains}
\end{equation}

Explicitly, $K_\star$ and $L_\star$ can be written as
\begin{align}
    K_* &= \Big(R_u + B_1^\top P_\star B_1 \nonumber\\
    &\quad - B_1^\top P_\star B_2 
    \left(-R_v + B_2^\top P_\star B_2\right)^{-1} 
    B_2^\top P_\star B_1\Big)^{-1} \nonumber \\
    &\quad \times \Big(B_1^\top P_\star A \nonumber\\
    &\quad - B_1^\top P_\star B_2 
    \left(-R_v + B_2^\top P_\star B_2\right)^{-1} 
    B_2^\top P_\star A\Big),
    \label{eq:K-star} \\[6pt]
    L_* &= \Big({-R_v + B_2^\top P_\star B_2} \nonumber\\
    &\quad - B_2^\top P_\star B_1 
    \left(R_u + B_1^\top P_\star B_1\right)^{-1} 
    B_1^\top P_\star B_2\Big)^{-1} \nonumber \\
    &\quad \times \Big(B_2^\top P_\star A \nonumber\\
    &\quad - B_2^\top P_\star B_1 
    \left(R_u + B_1^\top P_\star B_1\right)^{-1} 
    B_1^\top P_\star A\Big).
    \label{eq:L-star}
\end{align}

In our framework, we impose the following assumption.
\begin{assumption}\label{assm:game}
The following conditions hold: 
\emph{(i)} $(A, [B_1, B_2])$ is stabilizable and $(A, Q^{1/2})$ is detectable; 
\emph{(ii)} $R_v - B_2^\top P_\star B_2 \succ 0$, where $P_\star$ is the 
minimal positive definite solution to the GARE \eqref{eq:GARE}.
\end{assumption}

When the true system matrices $(A, B_1, B_2)$ are known, the optimal policy pair \eqref{eq:NE-gains} can be computed from the solution to \eqref{eq:GARE}. We consider the setting where the dynamics are unknown and aim to learn the equilibrium policies online from observed data. The performance of the learning algorithm is measured through the cumulative regret
\begin{equation}
\mathrm{Reg}(T)
\;:=\;
\sum_{t=0}^{T-1}
\mathbb{E}\!\bigl[c_t(x_t,u_t,v_t)\bigr]
\;-\;
T\, J_\star,
\label{eq:regret_def}
\end{equation}
and our objective is to achieve sublinear growth in $T$. This notion of regret compares the cumulative cost incurred by the learning algorithm with the steady optimal cost, and therefore quantifies the performance loss caused by learning under unknown dynamics. To this end, we develop a certified online learning algorithm in Section \ref{sec:algorithm} and establish a regret bound $O(\sqrt{T})$ in Section \ref{sec:regret}.

\section{Algorithm Design}
\label{sec:algorithm}
In this section, we develop an online learning framework for the zero-sum LQ game under unknown dynamics.
A natural approach is certainty equivalence \cite{mania2019certainty}: estimate the system parameters $(A,B_1,B_2)$ from data, solve the GARE
for the estimated model, and apply the resulting feedback gain. The algorithm is summarized in Algorithm~\ref{alg:ccepi}.

\subsection{Parameter Estimation and Confidence Bounds}
\label{subsec:estimation}
Since the system matrices are unknown, the first step is to estimate them from observed state and input data through ridge regression, while quantifying the estimation uncertainty with a confidence set.

We define the parameter matrix
$\Theta_\star := [A\;\;B_1\;\;B_2]\in\mathbb R^{n\times d}$, where $d= n+m_1+m_2$. The dynamics can be written equivalently as 

\begin{equation*}
    x_{t+1}=\Theta_\star z_t + w_{t},
\label{system_theta}
\end{equation*}
where $z_t=[x_t^\top\;u_t^\top\;v_t^\top]^\top$ is the regressor vector \cite{ioannou1996robust}.

We collect data $\{(z_t,x_{t+1})\}$ at episode $k$, and define the regularized design matrix $V_k := \lambda\, I_d + \sum_{t=0}^{t_k - 1} z_t\, z_t^\top$. Here, $t_k$ denotes the starting time of episode $k$. After we vectorize the parameter $\theta_\star := \mathrm{vec}(\Theta_\star)\in\mathbb R^{nd}$ and define the regressor matrix $\Phi_t=z_t^\top\otimes I_n \in \mathbb R^{n\times nd}$,
we can write the system as
\begin{equation}
    x_{t+1}=\Phi_t\theta_\star + w_{t},
\label{linear regression}
\end{equation}
and the associated block design matrix as
\begin{equation}
\mathbf V_k
:=
I_n\otimes V_k
=
\lambda I_{nd}
+
\sum_{t=0}^{t_k-1}
(z_t z_t^\top)\otimes I_n.
\label{eq:block_design}
\end{equation}
We then compute the ridge estimate by solving 
\begin{equation}
\widehat{\theta}_k
:= \arg\min_{\theta \in \mathbb{R}^{nd}}\;
\sum_{t=0}^{t_k - 1} \|x_{t+1} - \Phi_t\theta\|^2
+ \lambda\, \|\theta\|_2^2,
\label{eq:theta_obj}
\end{equation}
where $\lambda > 0$ is a regularization parameter. 
The closed-form solution of \eqref{eq:theta_obj} is $\widehat{\theta}_k
= \mathbf V_k^{-1}\sum_{t=0}^{t_k-1}\left( z_t\otimes I_n \right)x_{t+1}$.

Having obtained $\widehat{\theta}_k$, we now quantify its uncertainty from a high-probability confidence ellipsoid. Because \eqref{linear regression} is a linear regression model in $\theta_\star$ with $\Phi_t$, a standard self-normalized martingale concentration result for a regularized linear regression applies \cite{abbasi2011improved}.
Assume a known bound \(\|\theta_\star\|_2\le S_\theta\). The following lemma constructs a confidence ellipsoid for $\theta_\star$.

\begin{lemma}
\label{lemma:confidence}
In the presence of disturbance term in \eqref{eq:dynamics}, with probability at least $1-\delta$, we have
$\theta_\star \in \mathcal{C}_k(\delta)$ for all episodes $k \ge 1$,
where
\begin{equation}
\mathcal{C}_k(\delta)
\;:=\;
\left\{\theta \in \mathbb{R}^{n d}:\;
\|\theta - \widehat{\theta}_k\|_{\mathbf V_k} \;\le\; \beta_k(\delta)\right\},
\label{eq:Ck}
\end{equation}
with the radius of the confidence ellipsoid
\begin{equation}
\beta_k(\delta)
=
\sigma_w\sqrt{2n\log\!\left(
\frac{\det(V_k)^{1/2}}{\det(\lambda I_d)^{1/2}}
\right)
+
2\log\!\frac{1}{\delta}}
+\sqrt{\lambda}\,S_\theta.
\label{eq:beta_vec_simplified}
\end{equation}
\end{lemma}

\begin{proof}
Applying the self-normalized martingale bound to \eqref{linear regression}, with probability at least $1-\delta$, yields 
\begin{equation*}
\|\widehat{\theta}_k-\theta_\star\|_{\mathbf V_k}
\le
\sigma_w\sqrt{2\log\!\left(
\frac{\det(\mathbf V_k)^{1/2}}{\det(\lambda I_{nd})^{1/2}}
\cdot\frac{1}{\delta}
\right)}
+\sqrt{\lambda}\,\|\theta_\star\|_2.
\end{equation*}
Since $\mathbf V_k=I_n\otimes V_k$, we have $\det(\mathbf V_k)=\det(V_k)^n$ and $\det(\lambda I_{nd})=\det(\lambda I_d)^n$. Substituting these identities and using $\|\theta_\star\|_2\le S_\theta$ gives \eqref{eq:beta_vec_simplified}, which concludes the proof.
\end{proof}

\begin{algorithm}[t]
\caption{Certified Online Learning Algorithm}
\label{alg:ccepi}
\begin{algorithmic}[1]
\REQUIRE Update times $\mathcal T=\{t_k\}_{k\ge0}$, regularization $\lambda>0$, confidence level $\delta\in(0,1)$, regularity margins $(\mu,\gamma)$, initial certified model $\widetilde{\theta}_0\in\theta_{\mathrm{reg}}(\mu,\gamma)$, initial gain $K_0$,$L_0$, exploration sequence $\{\eta_t,\zeta_t\}_{t\ge0}$.
\STATE Initialize $V_0 \leftarrow \lambda I_d$, $K \leftarrow K_0$, $L \leftarrow L_0$, $k \leftarrow 0$.
\FOR{$t=0,1,2,\ldots$}
    \STATE Observe the current state $x_t$.
    \STATE Apply $u_t$ and $v_t$, and observe the $x_{t+1}$.
    \STATE Form the data pair $(z_t,x_{t+1}),$
    and update the information matrix $V_{t+1}$.
    \IF{$t+1=t_{k+1}$}
        \STATE Using all data $\{(z_\tau,x_{\tau+1})\}_{\tau=0}^{t}$, compute $\widehat{\theta}_{k+1}$
        and construct
        $\mathcal C_{k+1}(\delta)$
        \STATE Compute $\alpha_{k+1}\in[0,1]$ and set
        \[
        \widetilde{\theta}_{k+1}
        \leftarrow
        (1-\alpha_{k+1})\widetilde{\theta}_k+\alpha_{k+1}\widehat{\theta}_{k+1}.
        \]
        \STATE Solve the GARE at $\widetilde{\theta}_{k+1}$ to obtain $(K_{k+1},L_{k+1})$.
        \STATE Update $K \leftarrow K_{k+1}$, $L \leftarrow L_{k+1}$ and set $k \leftarrow k+1$.
    \ENDIF
\ENDFOR
\end{algorithmic}
\end{algorithm}

Lemma \ref{lemma:confidence} provides a high-probability confidence set for the unknown parameter $\theta_\star$ centered at the estimate $\widehat{\theta}_k$. However, it does not guarantee that $\widehat{\theta}_k$ itself lies in the parameter region where the GARE admits a stabilizing saddle-point solution, so Assumption \ref{assm:game} (ii) may not be satisfied. The next subsection introduces a certified surrogate selection step that resolves this difficulty.

\subsection{Certified Surrogate Model Selection}
\label{subsec:safe_model}
To ensure every policy update is feasible, instead of solving the GARE on $\widehat{\theta}_k$, we next introduce a certified surrogate selection step that selects a regularized model $\widetilde{\theta}_k$ inside the confidence set that
satisfies explicit GARE feasibility checks. We formalize the required regularity as follows. 

Let $\theta_{\mathrm{reg}}(\mu,\gamma)$ denote the set of parameters $\theta\in\mathbb{R}^{nd}$, with corresponding matrices $(A,B_1,B_2)$ as defined in Section~\ref{subsec:estimation} such that

(i) the GARE admits a stabilizing solution $P(\theta)\succeq 0$;

(ii) the solvability margin holds, \begin{equation*}
R_v - B_2^\top P(\theta) B_2 \succeq \mu I;
\label{eq:margin_reg}
\end{equation*}

(iii) the induced closed-loop matrix $A_{\mathrm{cl}}(\theta):=A-B_1K(\theta)-B_2L(\theta)$ is Schur stable with spectral radius at most $1-\gamma$.

At episode k, let $\mathcal C_k(\delta)$ be the confidence set from Lemma~\ref{lemma:confidence}. We then set the deployed surrogate model to satisfy $\widetilde{\theta}_k\in \mathcal C_k(\delta)\cap \theta_{\mathrm{reg}}(\mu,\gamma)$. Here, $\widetilde{\theta}_k$ must remain inside the current confidence region while also lying in a parameter region where the GARE is solvable with explicit regularity margins.

In implementation, we use a constructive
shrinkage procedure. At update $k$, given
the current estimate $\widehat{\theta}_k$ and the previous certified surrogate
$\widetilde{\theta}_{k-1}$, we search along the segment joining these two points
\begin{equation}
\theta_k(\alpha_k)
:=
(1-\alpha_k)\widetilde{\theta}_{k-1}
+\alpha_k \widehat{\theta}_k,
\qquad \alpha_k\in[0,1],
\label{eq:segment_search_vec}
\end{equation}

where $\alpha_k
:=
\sup\Bigl\{
\alpha\in[0,1]:
\theta_k(\alpha)\in
\mathcal C_k(\delta)\cap\theta_{\mathrm{reg}}(\mu,\gamma)
\Bigr\}$.
Since
$\widetilde{\theta}_k\in \mathcal C_k(\delta)$ and
$\theta_\star\in \mathcal C_k(\delta)$ on the high-probability event of
Lemma~\ref{lemma:confidence}, the surrogate remains uniformly close to the
true parameter. At the same time,
$\widetilde{\theta}_k\in\theta_{\mathrm{reg}}(\mu,\gamma)$ ensures that every
GARE solve used by the algorithm admits a stabilizing saddle-point solution
with the prescribed margins.

Having obtained $\widetilde{\theta}_k$, we next convert the ellipsoidal confidence bound of Lemma~\ref{lemma:confidence} into an $\ell_2$ rate, which the regret analysis in Section~\ref{sec:regret} requires.

\begin{assumption}
\label{assum:excitation}
There exist constants $\underline{\nu}>0$ and $k_0\ge 1$ such that, on the
high-probability event of Lemma~\ref{lemma:confidence},
\begin{equation}
\lambda_{\min}(V_k)\ge \underline{\nu}\, t_k,
\qquad \forall k\ge k_0.
\label{eq:info_growth}
\end{equation}
\end{assumption}
Assumption~\ref{assum:excitation} is a persistent-excitation requirement on the
full regressor $z_t$. In the zero-sum setting, 
sufficient variation must be present in the state and realized action for both players.

In practice, this may be enforced by injecting
bounded excitation through suitable perturbation signals together with a nondegeneracy condition on the induced adversary sequence.

\begin{lemma}\label{lemma:lemma2}
Suppose that the condition in Lemma~\ref{lemma:confidence} and Assumption~\ref{assum:excitation} holds, there exists $k_0 > 1$ such that for all $k\ge k_0$, the surrogate ridge estimate $\widetilde \theta_k$ satisfies the bound
\begin{equation}
\|\widetilde{\theta}_k-\theta_\star\|_2
\le
\frac{2\beta_k(\delta)}{\sqrt{\underline{\nu}\,t_k}}.
\label{eq:theta_l2_linear}
\end{equation}
\end{lemma}

\begin{proof}
On the high probability event of Lemma~\ref{lemma:confidence}, the true parameter satisfies $\|\widehat \theta_k-\theta_\star\|_{\mathbf V_k} \leq \beta_k(\delta)$. Since $\mathbf V_k$ is positive definite, we have $\|\widehat \theta_k-\theta_\star\|^2_{\mathbf V_k}\ge \lambda_{\min}(\mathbf V_k)\|\widehat \theta_k-\theta_\star\|_2^2$. Therefore, under Assumption~\ref{assum:excitation}, there exist $k_0 \ge 1$ such that $\forall k\ge k_0$, 
\begin{equation*}
\|\widehat \theta_k-\theta_\star\|_2\le\frac{\beta(\delta)}{\sqrt{\lambda_{\min}(\mathbf V_k)}} = \frac{\beta_k(\delta)}{\sqrt{\lambda_{\min}(V_k)}}\le \frac{\beta_k(\delta)}{\sqrt{\underline{\nu}\,t_k}}.
\end{equation*}
This bounds the estimation error of the ridge estimate $\widehat \theta_k$. Since the surrogate $\widetilde\theta_k$ is chosen from the confidence set also centered at $\widehat \theta_k$, the shrinkage step also satisfies a comparable bound such that 
\begin{equation*}
\|\widetilde \theta_k-\widehat \theta_k\|_2\le\frac{\beta_k(\delta)}{\sqrt{\underline{\nu}\,t_k}}.
\end{equation*}
Thus, by the triangle inequality, we can derive an estimation error bound on $\tilde \theta_k$ for some $k_0$:
\begin{equation}
\|\widetilde{\theta}_k-\theta_\star\|_2
\le
\frac{2\beta_k(\delta)}{\sqrt{\underline{\nu}\,t_k}},
\qquad \forall k\ge k_0.
\end{equation}
This concludes the proof.
\end{proof}

\subsection{Certified Certainty-Equivalent Control Update}
\label{subsec:control_update}

Given the certified surrogate model $\widetilde{\theta}_k$, we solve the GARE
associated with $\widetilde{\theta}_k$ and obtain the corresponding stabilizing
saddle-point gains $(K_k,L_k)$. The minimizer then applies
\begin{align}
  u_t &= -K_k x_t + \eta_t, \\
  v_t &= -L_k x_t + \zeta_t,
\end{align}
for $t = t_k, \ldots, t_{k+1}-1$, where $\eta_t, \zeta_t$ are exploration signals chosen so that the resulting regressor
sequence 
$z_t$ satisfies the information-growth
condition in Assumption~\ref{assum:excitation}. 

The controller is updated only when the design matrix $V_t$ has grown sufficiently, as measured by a doubling trick condition~\cite{abbasi2011regret}: 
\begin{equation}
t_{k+1}
:=
\min\left\{
t>t_k:\;
\det\bigl(V_t\big)\ge 2\,\det\bigl(V_{t_k}\bigr)
\right\}.
\label{eq:episode}
\end{equation}
This prevents unnecessary re-estimation when the newly collected data do not significantly improve identification. As a result, the number of updates up to time $T$ is at most logarithmic in $\det(V_T)$, \emph{i.e.}, $O\!\bigl(\log\det (V_T)\bigr)$, by the same determinant growth argument used in Lemma~8 of \cite{abbasi2011regret}.

\subsection{Local Perturbation of the GARE Saddle Solution}
\label{subsec:perturbation_main}
In this part, we present the key regularity statement, which will be used later in the regret analysis.
It shows that small model error induces proportionally small perturbations in the Riccati solution
and in the associated saddle gains.
For a given $\theta$ in a neighborhood of $\theta_\star$, let $P(\theta)$ denote the stabilizing solution to the GARE with system matrices $(A,B_1,B_2)$ extracted from $\theta$. To study how the $P$ varies with $\theta$, define the GARE residual map
\begin{equation}
  F(P,\theta) := P - \Phi(P,\theta),
  \label{eq:F_def}
\end{equation}
where $\Phi(P,\theta) := Q + A_{\mathrm{cl}}^\top P\, A_{\mathrm{cl}}
+ K^\top R_u K - L^\top R_v L$
and $A_{\mathrm{cl}} := A - B_1 K - B_2 L$ with $(K,L)$ being the stage saddle pair. Then $P(\theta)$ is characterized locally by $F(P(\theta),\theta)=0$.

\begin{theorem}
\label{thm:lipschitz_main}
There exists a neighborhood \(U\) of
\(\theta_\star\) and finite constants \(C_P,C_K,C_L>0\) such that, for every \(\theta\in U\),
\begin{equation}
\|P(\theta)-P_\star\|_F
\le
C_P\|\theta-\theta_\star\|_2,
\label{eq:P_lipschitz_main}
\end{equation}
and
\begin{equation}
\begin{gathered}
\|K(\theta)-K_\star\|_F
\le
C_K\|\theta-\theta_\star\|_2, \\
\|L(\theta)-L_\star\|_F
\le
C_L\|\theta-\theta_\star\|_2.
\end{gathered}
\label{eq:KL_lipschitz_main}
\end{equation}
Moreover, the maps
\(\theta\mapsto P(\theta)\),
\(\theta\mapsto K(\theta)\), and
\(\theta\mapsto L(\theta)\)
are continuously differentiable on \(U\).
\end{theorem}

\begin{proof}
By the envelope theorem, the stationarity conditions on the
saddle gains eliminate their derivatives from $D_PF$, so that
$D_PF(P_\star,\theta_\star)$ reduces to the Lyapunov operator $\mathcal{L}_\star(X) = X - A_{{\rm cl},\star}^\top X\, A_{{\rm cl},\star}$.
Because $A_{{\rm cl},\star}$ is Schur stable, $\mathcal{L}_\star$ is invertible.
Since \(F\) is smooth on the strong regularity set, the implicit function theorem yields a local
$C^1$ map $\theta \mapsto P(\theta)$ satisfying $F(P(\theta),\theta)=0$.
A uniform bound on $(D_P F)^{-1}$
and on $D_\theta F$ over a small neighborhood of $\theta_\star$ then implies a uniform bound on \(D_\theta P(\theta)\), and hence gives
\eqref{eq:P_lipschitz_main}.
The saddle gains $(K(\theta), L(\theta))$ are determined by the first-order optimality conditions of the stage game, which define them as smooth functions of $(P,\theta)$. Since $\theta\mapsto P(\theta)$ is $C^1$, so are
$\theta\mapsto K(\theta)$ and $\theta\mapsto L(\theta)$,
giving \eqref{eq:KL_lipschitz_main}. The full proof is given in
Appendix~\ref{app:lipschitz_proof}.
\end{proof}

Let \(U\) and \(C_P,C_K,C_L\) be as in Theorem~\ref{thm:lipschitz_main}.
If \(\widetilde{\theta}_k\in U\), then
\begin{equation}
\|P(\widetilde{\theta}_k)-P_\star\|_F
\le
C_P\|\widetilde{\theta}_k-\theta_\star\|_2,
\label{eq:corP_main}
\end{equation}
and
\begin{equation}
\begin{gathered}
\|K(\widetilde{\theta}_k)-K_\star\|_F
\le
C_K\|\widetilde{\theta}_k-\theta_\star\|_2, \\
\|L(\widetilde{\theta}_k)-L_\star\|_F
\le
C_L\|\widetilde{\theta}_k-\theta_\star\|_2.
\end{gathered}
\label{eq:corKL_main}
\end{equation}

\section{Convergence Analysis}
\label{sec:regret}

In this section, we analyze the performance by establishing the regret guarantee for Algorithm~\ref{alg:ccepi}. We provide a regret decomposition into transient, exploration, and policy gap terms and bound each term separately.
 
\subsection{Saddle-Point Cost Perturbation}

To characterize how deviations of the gains from the equilibrium pair $(K_\star, L_\star)$ affect the closed-loop average cost, we begin by controlling the policy gap term through a local perturbation bound on the saddle point cost around the equilibrium gains.
 
For stabilizing gains $(K,L)$ applied under the true dynamics
$\Theta_\star$, we define the closed-loop average cost
\begin{equation}
J(K,L)
\;:=\;
\mathrm{tr}\!\bigl(P_{K,L}\,\Sigma_w\bigr),
\label{eq:JKL}
\end{equation}
where $P_{K,L}\succeq 0$ solves the closed-loop Lyapunov equation
\begin{equation}
P_{K,L}
=
Q + K^\top R_u K - L^\top R_v L
+
A_{\mathrm{cl}}^\top P_{K,L}\, A_{\mathrm{cl}},
\label{eq:lyap_KL}
\end{equation}
with $A_{\mathrm{cl}}=A-B_1K-B_2L$. 

\begin{lemma}
\label{lemma:quadratic_cost}
    There exists a neighborhood $\mathcal{N}$ of $(K_\star,L_\star)$ and a
constant $C_J>0$ such that, for all $(K,L)\in\mathcal{N}$ with
$\rho(A_{\mathrm{cl}})<1$,
\begin{equation}
\bigl|J(K,L)-J_\star\bigr|
\;\le\;
C_J\!\left(\|\Delta K\|_F^2+\|\Delta L\|_F^2\right),
\label{eq:cost_gap_quadratic}
\end{equation}
where $\Delta K:=K-K_\star$, $\Delta L:=L-L_\star$.
\end{lemma}

Lemma~\ref{lemma:quadratic_cost} shows that in a neighborhood of the saddle point, the deviation of the average cost in \eqref{eq:JKL} is locally quadratic in the gain perturbations. The detailed proof is given in Appendix~\ref{app:cost_gap_proof}.

\subsection{Regret Decomposition and Bounding}

With the local cost perturbation result established above, in this subsection, we combine it with the episodic update structure of Algorithm~\ref{alg:ccepi} to derive the regret decomposition and bound each resulting term.

Let $\mathcal{K}_T$ denote the number of completed episodes up to time $T$, with
$\mathcal{E}_k=\{t_k,\dots,t_{k+1}-1\}$ and $T_k=t_{k+1}-t_k$.

On the event of Lemma~\ref{lemma:confidence},
\begin{equation}
\begin{aligned}
\bigl|\mathrm{Reg}(T)\bigr|
&\le
C_{\mathrm{tr}}\,\mathcal{K}_T
\;+\;
C_{\mathrm{ex}}\!\sum_{t=0}^{T-1}
\mathbb{E}\bigl[\|\eta_t\|^2 + \|\zeta_t\|^2\bigr] \\
& \qquad+
\sum_{k=0}^{\mathcal{K}_T}
T_k\bigl|J(K_k,L_k)-J_\star\bigr|,
\label{eq:regret_decomp}
\end{aligned}
\end{equation}
where $C_{\mathrm{tr}}$ depends on the spectral gap $\gamma$, and the
cost norms, $C_\mathrm{ex}$ depends on $\|R_u\|$ and $\|R_v\|$, and the steady-state
covariance bound.
 
The decomposition follows by partitioning the regret into episodes,
separating transient mixing costs from steady-state excess within each
episode, and isolating the exploration injection. We now bound the three terms in turn.
 
We first bound the term in \eqref{eq:regret_decomp}. For the transient cost $C_{\mathrm{tr}}\,\mathcal{K}_T$, each controller update at an episode boundary causes a cost bounded by a constant $C_{\mathrm{tr}}$ because $A_{\mathrm{cl}}$ is Schur stable. According to \cite{abbasi2011regret}, doubling trick \eqref{eq:episode} then ensures $\mathcal{K}_T = O(d\log T)$, giving
\begin{equation}
  C_{\mathrm{tr}}\,\mathcal{K}_T = O(d\log T).
  \label{eq:term_I}
\end{equation}
For the exploration cost $C_{\mathrm{ex}}\!\sum_{t=0}^{T-1}
\mathbb{E}\bigl[\|\eta_t\|^2 + \|\zeta_t\|^2\bigr]$, both players inject i.i.d.\ perturbations $\eta_t\sim\mathcal{N}(0,\sigma_\eta^2 I_{m_1})$ and
$\zeta_t\sim\mathcal{N}(0,\sigma_\zeta^2 I_{m_2})$ with
$\sigma_\eta^2=\sigma_\zeta^2=T^{-1/2}$ \cite{dean2018regret}, which yields
\begin{equation}
  C_{\mathrm{ex}}\!\sum_{t=0}^{T-1}
  \mathbb{E}\bigl[\|\eta_t\|^2+\|\zeta_t\|^2\bigr]
  = C_{\mathrm{ex}}\,(m_1+m_2)\sqrt{T}.
  \label{eq:term_II}
\end{equation}
For the policy gap $\sum_{k=0}^{\mathcal{K}_T}
T_k\bigl|J(K_k,L_k)-J_\star\bigr|$, chaining \eqref{eq:cost_gap_quadratic} with the gain Lipschitz bounds of \eqref{eq:corKL_main} leads to
\begin{equation}
\bigl|J(K_k,L_k)-J_\star\bigr|
\;\le\;
C_J(C_K^2+C_L^2)\,
\bigl\|\widetilde{\theta}_k-\theta_\star\bigr\|_2^2.
\label{eq:cost_gap_chain}
\end{equation}
On the event of Lemma~\ref{lemma:confidence} and under
Assumption~\ref{assum:excitation},
$\|\widetilde{\theta}_k-\theta_\star\|_2^2
\le 4\beta_k(\delta)^2/(\underline{\nu}\,t_k)$.
Since $\beta_k(\delta)^2=O(nd\log T)$ and the doubling rule gives
$T_k\le t_k$, we have
\begin{align}
  &\sum_{k=k_0}^{\mathcal{K}_T} T_k\bigl|J(K_k,L_k)-J_\star\bigr| \notag\\
  &\quad\;\le\;
  \frac{4C_J(C_K^2+C_L^2)}{\underline{\nu}}
  \sum_{k=k_0}^{\mathcal{K}_T}\frac{T_k\,\beta_k^2}{t_k}
  \;=\;
  O(nd^2\log^2 T).
  \label{eq:term_III}
\end{align}
 
\subsection{Main Result}
 
\begin{theorem}
\label{thm:regret_main}
With probability at least $1-\delta$,
Algorithm~\ref{alg:ccepi} with $\sigma_\eta^2=T^{-1/2}$ and doubling episodes achieves
\begin{equation}
{\;
\mathrm{Reg}(T)
\;=\;
{O}\!\left(\sqrt{T}\right),
\;}
\label{eq:regret_bound}
\end{equation}
\end{theorem}
 
\begin{proof}
Substituting \eqref{eq:term_I}--\eqref{eq:term_III} into \eqref{eq:regret_decomp}, for $T$ sufficiently large, which holds with probability at least
$1-\delta$,
\begin{align*}
\bigl|\mathrm{Reg}(T)\bigr|
&\le
{O(d\log T)}
\;+\;
{O\!\bigl(\sqrt{T}\bigr)}
\;+\;
{O(nd^2\log^2 T)} \\
&= O\!\bigl(\sqrt{T}\bigr).
\end{align*}
\end{proof}
 
\section{Simulation and Results}
\label{sec:simulation}

In this section, we provide numerical results to illustrate the performance of the proposed method in terms of system identification, policy learning, and regret. Specifically, the simulation is designed to highlight four aspects of the method: i) the model estimates, referring to the inferred system parameters computed from the data, improve over time, ii) the feedback gains, which are the control inputs for both players, approach the true equilibrium solution, iii) the certified surrogate, meaning the model adjusted to satisfy the regularity constraints, remains distinct from the raw estimate when necessary, iv) the regret decreases over time, indicating convergence of the learning process.

In this numerical setting, we consider a system with

\[
A=
\begin{bmatrix}
0.85 & 0.10 & 0.10\\
0.10 & 0.62 & 0.08\\
0.10 & 0.06 & 0.72
\end{bmatrix},
B_1=
\begin{bmatrix}
0.80\\
0.25\\
0.12
\end{bmatrix},
B_2=
\begin{bmatrix}
0.10\\
0.08\\
0.15
\end{bmatrix}.
\]

We set $Q=I_3$, $R_u = [1.1]$, and $R_v = [2.5]$.
The disturbance is Gaussian with covariance
$\Sigma_w = \sigma_w^2 I_3$ and $\sigma_w = 0.01$. The simulation runs through a horizon $T=50,000$, with the regularization parameter $\lambda=1$, the confidence failure probability $\delta=0.2$. The initial state is chosen as $x_0=[1.2,\,-0.90,\,0.70]^\top$. The initial candidate model is selected as a perturbed version of
the true parameter matrix and is accepted only if it
satisfies the regularity test. A benchmark value in the regret analysis is computed as $J_\star=\mathrm{tr}(P_\star \Sigma_w)$ where $P_\star$ is the Riccati solution associated with the true system dynamics.

\begin{figure}[H]
    \centering
    \includegraphics[width=\linewidth]{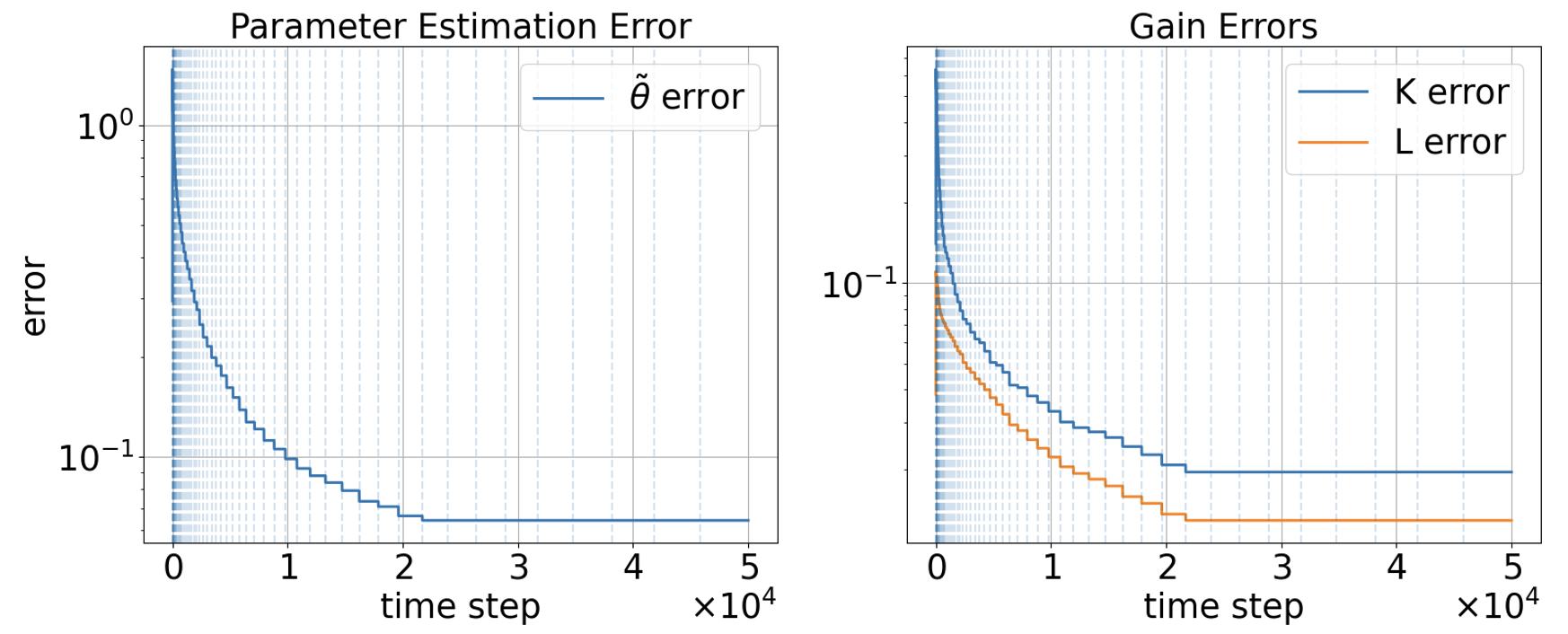}
    \caption{Theta estimation error (left) and gain error for both players (right).}
    \label{fig:error}
    \vspace{-0.2cm}
\end{figure}

Figure~\ref{fig:error} shows the convergence of the model estimation and the corresponding feedback gains. The left figure shows the error $\|\widetilde{\theta}_k-\theta_\star\|_2$ across update episodes, which indicates that the feasible model used for control moves progressively closer to the true dynamics as more data are collected. The right figure shows the gain errors $\|K_k-K_\star\|_F$ and $\|L_k-L_\star\|_F$. Both quantities decrease over time, showing that the computed saddle-point feedback approaches the true Nash gains. In Fig.~\ref{fig:error}, vertical lines indicate the time when the policy is updated, where both errors drop. This is consistent with the doubling schedule update rule used in Section~\ref{subsec:control_update}.

Figure~\ref{fig:shrinkage} shows the role of the model shrinkage step, and the observed regret behavior. The left figure compares the raw estimator error $\|\widehat{\theta}_k - \theta_\star\|_2$ with the certified model error $\|\widetilde{\theta}_k - \theta_\star\|_2$. The visible gap between the two curves reflects the process of selecting a more conservative model to guarantee the GARE solvability. The right figure shows the empirical regret normalized by $\sqrt{t}$. The curve remains bounded and gradually stabilizes, which is consistent with the sublinear regret result established in Theorem~\ref{thm:regret_main}.

\begin{figure}[t]
    \centering
    \includegraphics[width=\linewidth]{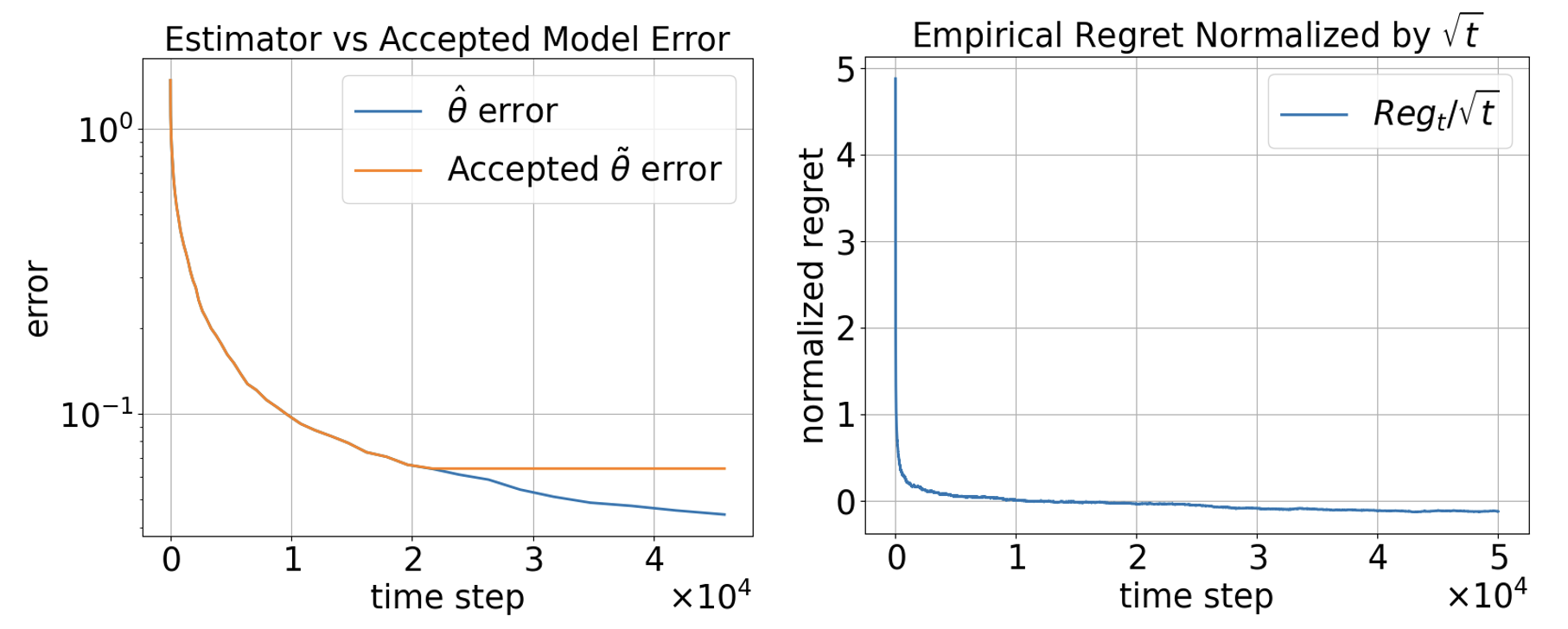}
    \caption{Comparison of $\widetilde \theta$ and $\widehat \theta$ (left), and regret convergence (right).}
    \label{fig:shrinkage}
    \vspace{-0.2cm}
\end{figure}

Overall, the simulation results verify the intended mechanism of the proposed algorithm, including parameter estimation, surrogate model shrinkage, updating gains for both players, and achieving stable learning performance. Although this numerical experiment is illustrative, it provides empirical evidence for the certified update rule and shows the potential of implementing it in a more complex but well-defined system.

\section{Concluding Remarks}
\label{sec:conclusion}

In this paper, we developed a certified online learning algorithm for a discrete-time LQ zero-sum game with unknown dynamics. We combined ridge least squares estimation, confidence set construction, and certified model selection to compute stabilizing saddle-point feedback policies. We showed that the selected surrogate model remains within a regular region that guarantees solvability of the GARE. We established a regret bound via a regret decomposition argument to demonstrate convergence of the algorithm. Numerical simulations further illustrated that the model estimates and feedback gains progressively converge to their true values, while the normalized regret remains bounded.

A natural direction for future research is to relax the persistent excitation assumption or provide a more rigorous justification for this condition. Extending the algorithm to partially observed games, to multi-player settings, or to time-varying systems also constitutes an important avenue for further investigation.



\bibliographystyle{IEEEtran} 
\bibliography{ref,ids}

@article{Xu2025when,
	author = {Xu, Gehui and Parisini, Thomas and Malikopoulos, Andreas A},
	journal = {arXiv preprint arXiv:2508.13450},
	title = {When Does Selfishness Align with Team Goals? {A} Structural Analysis of Equilibrium and Optimality},
	year = {2025}}

@article{Malikopoulos2024,
	author = {Malikopoulos, Andreas A.},
	journal = {European Journal of Control},
	number = {Part A},
	pages = {101043},
	title = {Combining Learning and Control in Linear Systems},
	volume = {80},
	year = {2024}}

@inproceedings{le2023multirobot,
	author = {Le, Viet-Anh and Chalaki, Behdad and Tadiparthi, Vaishnav and Mahjoub, Hossein Nourkhiz and D'sa, Jovin and Moradi-Pari, Ehsan and Malikopoulos, Andreas A},
	booktitle = {2024 IEEE International Conference on Robotics and Automation (ICRA)},
	organization = {IEEE},
	pages = {4834--4840},
	title = {{Multi-Robot Cooperative Navigation in Crowds: A Game-Theoretic Learning-Based Model Predictive Control Approach}},
	year = {2024}}

@article{Malikopoulos2022a,
	author = {Malikopoulos, Andreas A},
	journal = {Automatica},
	number = {110912},
	title = {Separation of Learning and Control for Cyber-Physical Systems},
	volume = {151},
	year = {2023}}

@article{chremos2022CSMArticle,
	author = {Chremos, Ioannis Vasileios and Malikopoulos, Andreas A},
	journal = {IEEE Control Systems},
	number = {1},
	pages = {20--45},
	title = {Mechanism Design Theory in Control Engineering: A Tutorial and Overview of Applications in Communication, Power Grid, Transportation, and Security Systems},
	volume = {44},
	year = {2024}}

@inproceedings{kounatidis2025combined,
	arxivid = {2510.00308},
	author = {Kounatidis, Panagiotis and Malikopoulos, Andreas A.},
	booktitle = {65th American Control Conference (ACC)},
	note = {in review},
	title = {Combined Learning and Control: A New Paradigm for Optimal Control with Unknown Dynamics},
	year = {2025}}

@article{abbasi2011improved,
  title={Improved algorithms for linear stochastic bandits},
  author={Abbasi-Yadkori, Yasin and P{\'a}l, D{\'a}vid and Szepesv{\'a}ri, Csaba},
  journal={Advances in neural information processing systems},
  volume={24},
  year={2011}
}

@book{bacsar2008h,
  title={H-infinity optimal control and related minimax design problems: a dynamic game approach},
  author={Ba{\c{s}}ar, Tamer and Bernhard, Pierre},
  year={2008},
  publisher={Springer Science \& Business Media}
}

@book{lewis2012optimal,
  title={Optimal control},
  author={Lewis, Frank L and Vrabie, Draguna and Syrmos, Vassilis L},
  year={2012},
  publisher={John Wiley \& Sons}
}

@article{dean2018regret,
  title={Regret bounds for robust adaptive control of the linear quadratic regulator},
  author={Dean, Sarah and Mania, Horia and Matni, Nikolai and Recht, Benjamin and Tu, Stephen},
  journal={Advances in Neural Information Processing Systems},
  volume={31},
  year={2018}
}

@article{mania2019certainty,
  title={Certainty equivalence is efficient for linear quadratic control},
  author={Mania, Horia and Tu, Stephen and Recht, Benjamin},
  journal={Advances in neural information processing systems},
  volume={32},
  year={2019}
}

@inproceedings{simchowitz2020naive,
  title={Naive exploration is optimal for online lqr},
  author={Simchowitz, Max and Foster, Dylan},
  booktitle={International Conference on Machine Learning},
  pages={8937--8948},
  year={2020},
  organization={PMLR}
}

@book{ioannou1996robust,
  title={Robust adaptive control},
  author={Ioannou, Petros A and Sun, Jing},
  volume={1},
  year={1996},
  publisher={PTR Prentice-Hall Upper Saddle River, NJ}
}

@article{zhang2019policy,
  title={Policy optimization provably converges to Nash equilibria in zero-sum linear quadratic games},
  author={Zhang, Kaiqing and Yang, Zhuoran and Basar, Tamer},
  journal={Advances in Neural Information Processing Systems},
  volume={32},
  year={2019}
}

@article{al2007model,
  title={Model-free Q-learning designs for linear discrete-time zero-sum games with application to H-infinity control},
  author={Al-Tamimi, Asma and Lewis, Frank L and Abu-Khalaf, Murad},
  journal={Automatica},
  volume={43},
  number={3},
  pages={473--481},
  year={2007},
  publisher={Elsevier}
}

@article{nortmann2024nash,
  title={Nash equilibria for linear quadratic discrete-time dynamic games via iterative and data-driven algorithms},
  author={Nortmann, Benita and Monti, Andrea and Sassano, Mario and Mylvaganam, Thulasi},
  journal={IEEE Transactions on Automatic Control},
  volume={69},
  number={10},
  pages={6561--6575},
  year={2024},
  publisher={IEEE}
}

@inproceedings{abbasi2011regret,
  title={Regret bounds for the adaptive control of linear quadratic systems},
  author={Abbasi-Yadkori, Yasin and Szepesv{\'a}ri, Csaba},
  booktitle={Proceedings of the 24th annual conference on learning theory},
  pages={1--26},
  year={2011},
  organization={JMLR Workshop and Conference Proceedings}
}

@ARTICLE{Ye2023gamesurvery,
  author={Ye, Maojiao and Han, Qing-Long and Ding, Lei and Xu, Shengyuan},
  journal={Proceedings of the IEEE}, 
  title={Distributed Nash Equilibrium Seeking in Games With Partial Decision Information: A Survey}, 
  year={2023},
  volume={111},
  number={2},
  pages={140-157},
  doi={10.1109/JPROC.2023.3234687}}

@article{ren2025chance,
  title={Chance-constrained linear quadratic gaussian games for multi-robot interaction under uncertainty},
  author={Ren, Kai and Salizzoni, Giulio and G{\"u}rsoy, Mustafa Emre and Kamgarpour, Maryam},
  journal={IEEE Control Systems Letters},
  year={2025},
  publisher={IEEE}
}

@article{xu2025game,
  title={Game-Theoretic Learning-Based Mitigation of Insider Threats},
  author={Xu, Gehui and Chen, Kaiwen and Parisini, Thomas and Malikopoulos, Andreas A},
  journal={arXiv preprint arXiv:2512.03222},
  year={2025}
}

@article{liu2023potential,
  title={Potential game-based decision-making for autonomous driving},
  author={Liu, Mushuang and Kolmanovsky, Ilya and Tseng, H Eric and Huang, Suzhou and Filev, Dimitar and Girard, Anouck},
  journal={IEEE Transactions on Intelligent Transportation Systems},
  volume={24},
  number={8},
  pages={8014--8027},
  year={2023},
  publisher={IEEE}
}

\section*{APPENDIX}
\subsection{Proof of Theorem~\ref{thm:lipschitz_main}}
\label{app:lipschitz_proof}
\begin{proof}
The proof proceeds in three steps:
(i)~we show that $D_P F(P_\star,\theta_\star)$ reduces to an invertible Lyapunov operator
and bound its inverse;
(ii)~we bound $D_\theta F$ and apply the implicit function theorem to
obtain the Lipschitz bound on $P(\theta)$;
(iii)~we transfer this bound to the gains $(K(\theta), L(\theta))$.

Throughout, $F$, $\Phi$, and $A_{\mathrm{cl}}$ are as defined in
Section~\ref{subsec:perturbation_main}.
For a given $(P,\theta)$ in the strong regularity region, we abbreviate
$K := K(P,\theta)$, $L := L(P,\theta)$,
and $A_{\mathrm{cl}} := A - B_1 K - B_2 L$.
We define the stage saddle functional as
\begin{align*}
\mathcal{J}(P,\theta,K,L)
&:= Q + K^\top R_u K - L^\top R_v L \notag 
 + A_{\mathrm{cl}}^\top PA_{\mathrm{cl}}.
\label{eq:J_appendix}
\end{align*}
By construction, $\Phi(P,\theta)=\mathcal{J}(P,\theta,K(P,\theta),L(P,\theta)).$ For any direction \(X\in\mathbb{R}^{n\times n}\), applying the chain rule,
\begin{align*}
D_P\Phi(P,\theta)[X]
&=
D_P\mathcal{J}[X]
+
D_K\mathcal{J}\big[D_PK[X]\big]
\notag\\
&\quad
+
D_L\mathcal{J}\big[D_PL[X]\big].
\end{align*}

Since \((K,L)\) is the stage saddle pair, the first-order stationarity conditions with respect to $(K,L)$ imply $D_K\mathcal{J}=0$ and $D_L\mathcal{J}=0$, which leads to
\begin{equation*}
D_P\Phi(P,\theta)[X]
=
D_P\mathcal{J}[X]
=
A_{\rm cl}^\top X A_{\rm cl}.
\end{equation*}

Since \(F(P,\theta)=P-\Phi(P,\theta)\), we obtain $D_PF(P,\theta)[X]
=
X-A_{\rm cl}^\top X A_{\rm cl}.$ At $(P_\star,\theta_\star)$, this becomes the Lyapunov operator
$\mathcal{L}_\star(X)
:=
X-A_{{\rm cl},\star}^\top X A_{{\rm cl},\star}.$ which is invertible because $A_{{\rm cl},\star}$ is Schur, with inverse
\begin{equation}
\mathcal{L}_\star^{-1}(X)
=
\sum_{t=0}^{\infty}
(A_{{\rm cl},\star}^\top)^t X (A_{{\rm cl},\star})^t.
\label{eq:lyap_inverse_appendix}
\end{equation}

To bound $\|\mathcal{L}_\star^{-1}\|_{F \to F}$, let
$S_\star \succ 0$ satisfy
$I \preceq S_\star \preceq \kappa I$ and
$A_{\mathrm{cl},\star}^\top S_\star\, A_{\mathrm{cl},\star}
\preceq (1-\gamma)^2 S_\star$,
and define the weighted Frobenius norm
$\|X\|_{S_\star} := \|S_\star^{1/2} X S_\star^{1/2}\|_F$.
Then $\|A_{\mathrm{cl},\star}^t\|_{S_\star \to S_\star} \le (1-\gamma)^t$
for all $t \ge 0$, and summing the geometric series in
\eqref{eq:lyap_inverse_appendix} gives
$\|\mathcal{L}_\star^{-1}(X)\|_{S_\star}
\le \frac{1}{\gamma(2-\gamma)} \|X\|_{S_\star}$. Since $I\preceq S_\star\preceq \kappa I$, the weighted and unweighted norms $\|X\|_F \le \|X\|_{S_\star} \le \kappa \|X\|_F$, and therefore
\begin{equation}
\|\mathcal{L}_\star^{-1}\|_{F\to F}
\le
\frac{\kappa}{\gamma(2-\gamma)}.
\label{eq:unweighted_inverse_bound_appendix}
\end{equation}

For a parameter perturbation $\delta\theta$, let $(\delta A,\delta B_1,\delta B_2)$ denote the corresponding block perturbations and define $\delta A_{\rm cl}:=\delta A-\delta B_1K-\delta B_2L.$
Using the same envelope argument and the gain derivatives from $D_\theta\Phi$, we obtain
\begin{equation}
D_\theta F(P,\theta)[\delta\theta]
=
-
\Bigl(
\delta A_{\rm cl}^\top P A_{\rm cl}
+
A_{\rm cl}^\top P\,\delta A_{\rm cl}
\Bigr).
\label{eq:DthetaF_appendix}
\end{equation}
With $\|\delta\theta\|_2^2 := \|\delta A\|_F^2 + \|\delta B_1\|_F^2 + \|\delta B_2\|_F^2$, applying Cauchy-Schwarz gives
\begin{align} 
\|D_\theta F(P,\theta)[\delta\theta]\|_F
&\le
2\|P\|\,\|A_{\mathrm{cl}}\|\,\|\delta A_{\mathrm{cl}}\|_F
\notag\\
&\le
2\|P\|\,\|A_{\rm cl}\|
\sqrt{1+\|K\|^2+\|L\|^2}\,
\|\delta\theta\|_2.
\label{eq:DthetaF_bound_appendix}
\end{align}

Since \(D_PF(P_\star,\theta_\star)=\mathcal{L}_\star\) is invertible and
\(F\) is smooth on the strong regularity set, the implicit function theorem implies that
there exist a neighborhood \(U\) of \(\theta_\star\), a neighborhood \(\mathcal{P}\) of \(P_\star\),
and a unique \(C^1\) map $P:U\to\mathcal{P}$ satisfying $F(P(\theta),\theta)=0, \forall\, \theta\in U.$
Differentiating this identity gives
\begin{equation}
D_\theta P(\theta)
=
-
\bigl(D_PF(P(\theta),\theta)\bigr)^{-1}
D_\theta F(P(\theta),\theta).
\label{eq:IFT_appendix}
\end{equation}

By continuity, after possibly shrinking \(U\), the inverse
\((D_PF(P(\theta),\theta))^{-1}\) remains uniformly bounded on \(U\).
Together with \eqref{eq:DthetaF_bound_appendix} and the continuity of
\(P(\theta)\), \(K(\theta)\), \(L(\theta)\), and \(A_{\rm cl}(\theta)\), this yields a finite constant
$C_P$ such that $\sup_{\theta \in U} \|D_\theta P(\theta)\|_{F \to F} \le C_P$. By mean-value theorem,
\[
\|P(\theta)-P_\star\|_F
\le
C_P\|\theta-\theta_\star\|_2,
\qquad \forall\, \theta\in U.
\]

The saddle gains are determined by the first-order optimality conditions
of the stage game, which define them as smooth functions of $(P,\theta)$
whenever $H(P)$ in~\eqref{eq:H} remains nonsingular.
Since $P(\theta)$ is $C^1$, and $H(P_\star)$ is nonsingular by Assumption \ref{assm:game}, $H(P(\theta),\theta)$ remains uniformly nonsingular on a sufficiently small neighborhood of
$\theta_\star$, the map $\theta\mapsto (K(\theta),L(\theta))$
is $C^1$ on $U$.
Therefore, it is locally Lipschitz, and there exist finite constants $C_K,C_L$ such that
\begin{align*}
\|K(\theta)-K_\star\|_F
&\le
C_K\|\theta-\theta_\star\|_2,
\\
\|L(\theta)-L_\star\|_F
&\le
C_L\|\theta-\theta_\star\|_2.
\end{align*}
This completes the proof.
\end{proof}

\subsection{Proof of Lemma~\ref{lemma:quadratic_cost}}
\label{app:cost_gap_proof}
\begin{proof}
Write $\Delta P:=P_{K,L}-P_\star$,
$\Delta K:=K-K_\star$, $\Delta L:=L-L_\star$, $A_{\mathrm{cl},\star}:=A-B_1K_\star-B_2L_\star$,
and $\delta B:=B_1\Delta K+B_2\Delta L$,
so that $A_{\mathrm{cl}}=A_{\mathrm{cl},\star}-\delta B$.

Subtracting the GARE \eqref{eq:GARE} for $P_\star$ from the Lyapunov
equation~\eqref{eq:lyap_KL} for $P_{K,L}$ and expanding $A_{\mathrm{cl}}=A_{\mathrm{cl},\star}-\delta B$, gives
\begin{align}
\label{eq:DeltaP_app}
\Delta P
&=
{A_{\mathrm{cl},\star}}^\top \Delta P\, A_{\mathrm{cl},\star}
\;+\;
\Delta K^\top R_u\, \Delta K \\
\notag&\quad -\;
\Delta L^\top R_v\, \Delta L
\;+\;
\delta B^\top P_\star\,\delta B
\;+\; R_1 \;+\; R_{\ge 3},
\end{align}
where $R_1$ collects all terms linear in $(\Delta K,\Delta L)$
and $R_{\ge 3}$ collects all remaining terms, each involving
$\Delta P$ multiplied by at least one factor of $\delta B$;
since $\Delta P$ is itself $O(\|\Delta K\|^2+\|\Delta L\|^2)$
by the preceding quadratic structure, these contribute at
cubic order and above. 
The GARE stationarity conditions yield
$B_1^\top P_\star A_{\mathrm{cl},\star} = R_u K_\star$ and
$B_2^\top P_\star A_{\mathrm{cl},\star} = -R_v L_\star$,
which can be verified by substituting~\eqref{eq:NE-gains} and
expanding $A_{\mathrm{cl},\star}$.
These identities cause every linear term in $R_1$ to cancel,
so $R_1=0$.

Applying the
$\mathcal{L}_\star^{-1}$ to \eqref{eq:DeltaP_app} and taking the trace against $\Sigma_w$,
\begin{align}
J(K,L)-J_\star
&=
\mathrm{tr}\!\bigl(
\mathcal{L}_\star^{-1}[\Delta K^\top R_u\,\Delta K]\,\Sigma_w
\bigr)
\notag\\
&\quad-\;
\mathrm{tr}\!\bigl(
\mathcal{L}_\star^{-1}[\Delta L^\top R_v\,\Delta L]\,\Sigma_w
\bigr)
\notag\\
&\quad+\;
\mathrm{tr}\!\bigl(
\mathcal{L}_\star^{-1}[\delta B^\top P_\star\,\delta B]\,\Sigma_w
\bigr)
\notag\\
&\quad+\;
O\!\bigl(\|\Delta K\|^3+\|\Delta L\|^3\bigr).
\label{eq:cost_expansion_app}
\end{align}
The first two terms are bounded by
$\|\mathcal{L}_\star^{-1}\|\,\|\Sigma_w\|$
times the Frobenius norm of the quadratic expression.
Expanding $\delta B^\top P_\star\,\delta B$ and
applying the AM--GM inequality to the cross terms
$\Delta K^\top B_1^\top P_\star B_2\,\Delta L$
shows this norm is at most
$C\bigl(\|\Delta K\|_F^2+\|\Delta L\|_F^2\bigr)$
for a constant $C$ depending on $\|R_u+B_1^\top P_\star B_1\|$,
$\|R_v-B_2^\top P_\star B_2\|$, and $\|B_1^\top P_\star B_2\|$.
In a sufficiently small neighborhood $\mathcal{N}$,
the cubic remainder is absorbed, giving~\eqref{eq:cost_gap_quadratic}.
\end{proof}

\end{document}